\documentclass{llncs}

\usepackage[lncs,amsthm,autoqed]{pamas}

\newcommand\eat[1]{}

\title{Complexity of coalition structure generation}
\author{%
	Haris Aziz\inst{1} \and
	Bart de Keijzer\inst{2}}

\institute{%
	Institut f\"ur Informatik,
	Technische Universit\"at M\"unchen, 
	80538 M\"unchen, Germany \\
	\email{aziz@in.tum.de}
	\and 
		CWI Amsterdam, 1098 XG Amsterdam, The Netherlands\\
	\email{B.de.Keijzer@cwi.nl}}

\begin{document}

\sloppy
\title{Complexity of coalition structure generation}
\maketitle

\begin{abstract}
We revisit the \emph{coalition structure generation problem} in which the goal is to partition the players into exhaustive and disjoint coalitions so as to maximize the social welfare. One of our key results is a general polynomial-time algorithm to solve the problem for all coalitional games provided that player types are known and the number of player types is bounded by a constant. As a corollary, we obtain a polynomial-time algorithm to compute an optimal partition for weighted voting games with a constant number of weight values and for coalitional skill games with a constant number of skills. We also consider well-studied and well-motivated coalitional games defined compactly on combinatorial domains. For these games, we characterize the complexity of computing an optimal coalition structure by presenting polynomial-time algorithms, approximation algorithms, or $\mathsf{NP}$-hardness and inapproximability lower bounds.
\end{abstract}

\section{Introduction}
Coalition formation is an important issue in multiagent systems with cooperating agents. 
Coalitional games have been used to model various cooperative settings in operations research, artificial intelligence and multiagent systems~\citep[see e.g, ][]{BaRo08a,BaRo09a,EGG+07a}. 
The area of coalitional game theory which studies coalition formation has seen considerable growth over the last few decades.
Given a set of agents $N$, a coalitional game is defined by a valuation function $v:N\rightarrow R$ where for $C \subseteq N$, $v(C)$ signifies the value which players in $C$ can generate by cooperating.

In a coalitional game, a partition of the players into exhaustive and disjoint coalitions is called a \emph{coalition structure}.
In the \emph{coalition structure generation problem}, the goal is to find a coalition structure $\pi$ of $N$ that maximizes the social welfare $\sum_{C\in \pi}v(C)$. 
We will refer to this problem of finding an optimal coalition structure as {\sc OptCS}.
In this paper, we conduct a detailed investigation of computing optimal coalition structures that give the maximum  social welfare. 
Computing optimal coalition structures is a natural problem in which the aim is to utilize resources in the most efficient manner. 
 
{\sc OptCS} has received attention in the artificial intelligence community where the focus has generally been on computing optimal coalition structures for general coalition formation games~\citep{MSRWBJ10a,SLA+99a} without any combinatorial structure. Traditionally, the input considered is an oracle called a characteristic function which returns the value for any given coalition (in time polynomial in the number of players). In this setting, it is generally assumed that the value of a coalition does not depend on players who are not in the coalition.
Computing optimal coalition structures is a computationally hard task because of the huge number of coalition structures. The total number of coalition structures for a player set of size $n$ is $B_n\sim \Theta(n^n)$ where $B_n$ is the $n$th Bell number. A number of algorithms have been developed in the last decade which attempt to satisfy many desirable criteria, e.g.\ outputting an optimal solution or a good approximation, the ability to prune, the anytime property, worst case guarantees, distributed computation etc.~\citep{MSRWBJ10a,RSJG09a,SLA+99a,SeAd10b}. In all of the cases,  the algorithms have a worst-case time complexity which is exponential in $n$. 
In this paper, we show that the picture is not that bleak if player types are known and the number of player types is bounded by a constant. In fact for such a condition, there is a polynomial-time algorithm for {\sc OptCS} for all coalitional games. In many multiagent systems, it can be reasonable to assume that the agents can be divided into a bounded number of types according to the player attributes.

We also study the complexity of {\sc OptCS} for a number of compact coalitional games.
Coalitional games can be represented compactly on combinatorial domains where the valuation function is implicitly defined~\citep{DeFa08a,DePa94a}. Numerous such classes of coalitional games have been the subject of recent research in multiagent systems: weighted voting games~\citep{EGG+07a}; skill games~\citep{BaRo08a}; multiple weighted voting games~\citep{ABH10a}; network flow games~\citep{BaRo09a}; spanning connectivity games~\citep{ALPS09b}; and matching games \citep{KePa03a}. 
Apart from some exceptions~(skill games~\citep{BMJK10a} and marginal contribution nets~\citep{0CI+09a}), most of the algorithmic research for these classes of games has been on computing stability-based solutions. In the paper, we characterize the complexity of {\sc OptCS} for many compact games by presenting polynomial-time exact algorithms, approximation algorithms, or $\mathsf{NP}$-hardness and inapproximability lower bounds. Throughout the paper, we assume familiarity with fundamental concepts in computational complexity~\citep{complexity}.

\paragraph{Contribution} 

In this paper, we undertake a detailed and systematic study of computing optimal coalition structures for many important combinatorial optimization coalitional games.

Our most important result is a general polynomial-time algorithm to compute an optimal coalition structure for any coalitional game when the player types are known and the number of player types is bounded by a fixed constant. 
As a corollary, we obtain a polynomial-time algorithm to compute an optimal coalition structure for weighted voting games with a constant number of weight values, linear games with a constant number of desirability classes, and all known coalitional skill games with a constant number of skills. 

In contrast to our general algorithmic result, we show that finding the player types is intractable in general from a communication and computational complexity point of view.

 We present a 2-approximation algorithm for the case of weighted voting games and show that this approximation bound is the best possible. 
	Our approximation and inapproximability results concerning weighted voting games may be of independent interest since they address a problem in the family of knapsack problems~\citep{KelPfePis04} which has not been studied before.
	
	We also examine well-known coalitional games based on graphs and characterize the complexity of computing the optimal coalition structures. 
	Interestingly for certain combinatorial optimization games for which the combinatorial optimization problem is $\mathsf{NP}$-hard, the problem of computing an optimal coalition structure is easy.


\section{Preliminaries}
\label{sec-prel}
In this section, we define several important classes of coalitional games and formally define the fundamental computational problem {\sc OptCS}.  

\subsection{Coalitional games}
We begin with the formal definition of a \emph{coalitional game}.

\begin{definition}[Coalitional games]	
A \emph{coalitional game} is a pair $(N,v)$ where $N=\{1,\ldots, n\}$ is a set of players and $v:2^N \rightarrow \mathbb{R}$ is a \emph{characteristic or valuation function} that associates with each coalition $C \subseteq N$ a payoff $v(C)$ where $v(\varnothing)=0$.
A coalitional game~$(N,v)$ is \emph{monotonic} when it satisfies the property that $v(C)\leq v(D)$ if $C \subseteq D$.
\end{definition}
Throughout the paper, when we refer to a coalitional game, we assume such a coalitional game with transferable utility. For the sake of brevity, we will sometimes refer to the game $(N,v)$ as simply~$v$.

\begin{definition}[Simple game]	
A \emph{simple game} is a monotonic coalitional game $(N,v)$ with $v:2^N \rightarrow \{0,1\}$ such that $v(\varnothing)=0$ and $v(N)=1$. A coalition $C \subseteq N$ is \emph{winning} if $v(C)=1$ and \emph{losing} if $v(C)=0$. A \emph{minimal winning coalition} (MWC) of a simple game $v$ is a winning coalition in which defection of any player makes the coalition losing.
A simple game can be represented by $(N,W^m)$, where $W^m$ is the set of minimal winning coalitions.
\end{definition}

For any monotonic coalitional game, one can construct a corresponding threshold game. Threshold versions are common in the multiagent systems literature; see for instance \citep{BaRo09a, EGG+07a}.
\begin{definition}[Threshold versions]\label{def:threshold}
For each coalitional game~$(N,v)$ and each threshold~$t\in\mathbb R^+$, the \emph{corresponding threshold game} is defined as the coalitional game~$(N,v^t)$, where
\begin{displaymath}
v^t(C)= \begin{cases} 1 & \textrm{if $v(C)\geq t$,}\\ 
0 & \textrm{otherwise.} \end{cases}%
\end{displaymath}
\end{definition}
It can easily be verified that if a game~$(N,v)$ is monotonic, then for any threshold~$t \leq v(N)$, the threshold version~$(N,v^t)$ is a simple game. 

\subsection{Coalitional game classes}
We now review a number of specific classes of coalitional games. Here we adopt the convention that if CLASS denotes a particular class of games, we have T-CLASS refer to the class of threshold games corresponding to games in CLASS, \ie for every threshold~$t$, $(N,v^t)$ is in T-CLASS if and only if $(N,v)$ is in CLASS.

Weighted voting games are a widely used class of monotonic games.
\begin{definition}[Weighted voting games~\citep{EGG+07a}]	
A \emph{weighted voting game (WVG)} is a simple game $(N,v)$ for which there is a \emph{quota $q \in \mathbb R^+$} and a \emph{weight~$w_i \in \mathbb R^+$} for each player~$i$  such that 
\[\text{
$v(C)=1$ if and only if $\sum_{i\in C}w_i\ge q$.
}
\]
The WVG with quota~$q$ and weights~$w_1,\ldots,w_n$ for the players is denoted by $[q;w_1,\ldots,w_n]$, where we commonly assume $w_i \geq w_{i+1}$  for $1\le i<n$. 

A \emph{multiple weighted voting game} (MWVG) is the simple game $(N,v)$ for which there are WVGs $(N,v_1),\ldots,(N,v_m)$  such that 
\[\text{
$v(C)=1$ if and only if $v_k(C)=1$ for $1 \leq k \leq m$.}
\]
We denote the MWVG game composed of $(N,v_1),\ldots,(N,v_m)$ by $(N,v_1\wedge\cdots\wedge v_m)$. 
\end{definition}
 
Other important classes of games are defined on graphs. Among these are \emph{spanning connectivity games}, \emph{independent set games}, \emph{matching games}, \emph{network flow games}, and \emph{graph games}, where either nodes or edges are controlled by players and the value of a coalition of players depends on their ability to connect the graph, enable a bigger flow, or obtain a heavier matching or edge set.
\begin{definition}[Spanning connectivity game \citep{ALPS09b}]
For each connected undirected graph $G = (V,E)$, we define the \emph{spanning connectivity game (SCG)} on $G$ as the simple game $(N,v)$ where $N=E$ and for all $C \subseteq E$, $v(C) = 1$ if and only if there exists some $E' \subseteq C$ such that $T = (V,E')$ is a spanning tree.
\end{definition}

\begin{definition}[Independent set game \citep{DeFa08a}]
For each connected undirected graph $G = (V,E)$, we define the \emph{independent set game (ISG)} on $G$ as the game $(N,v)$ where $N = V$ and for all $C \subseteq V$, $v(C)$ is cardinality of the maximum independent set on the subgraph of $G$ induced on $C$.
\end{definition}

\begin{definition}[Matching game \citep{KePa03a}]
Let $G = (V,E,w)$ be a weighted graph.
The  \emph{matching game} corresponding to~$G$ is the coalitional game $(N,v)$ with $N=V$ and for each $C \subseteq N$, the value $v(C)$ equals the weight of the maximum weighted matching of the subgraph induced by $C$. 
\end{definition}

\emph{Graph games} are likewise defined on weighted graphs \citep{DePa94a}.

\begin{definition}[Graph game \citep{DePa94a}]
For a weighted graph $(V,E,w)$, the \emph{graph game (GG)} is the coalitional game $(N,v)$ where $N=V$ and for $C \subseteq N$, $v(C)$ is the weight of edges in the subgraph induced by $C$. In this paper, we sometimes assume that the graph corresponding to a graph game has only positive edge weights and denote such graph games by GG$^+$. We denote the class of graph games where negative edge weights are allowed by GG. Note that for this latter general class of graph games, we allow the characteristic function $v$ to map to negative reals.
\end{definition}

A flow network  $(V,E,c,s,t)$ consists of a directed graph $(V,E)$, with capacity on edges $c:E \rightarrow \mathbb{R}^+$, a source vertex $s\in V$, and a sink vertex $t\in V$. A network flow is a function $f:E\rightarrow \mathbb{R}^+$, which obeys the capacity constraints and the condition that the total flow entering any vertex (other than $s$ and $t$) equals the total flow leaving the vertex. The value of the flow is the maximum amount flowing out of the source. 
\begin{definition}[Network flow game \citep{BaRo09a}]	
For a flow network $(V,E,c,s,t)$, the associated \emph{network flow game (NFG)} is the coalitional game $(N,v)$, where $N = E$ and for each $C \subseteq E$ the value $v(C)$ is the value of the maximum flow~$f$ with $f(e) = 0$ for all $e \in E \setminus C$. 
\end{definition} 

\begin{definition}[Path coalitional games]
	For an unweighted directed/undirected graph, $G=(V \cup \{s,t\}, E)$,
	\begin{itemize}
		\item the corresponding \emph{Edge Path Coalitional Game (EPCG)} is a simple coalitional game $(N,v)$ such that $N=E$ and for any $C \subseteq N$, $v(C)=1$ if and only if $C$ admits an $s$-$t$ path.
		\item the corresponding \emph{Vertex Path Coalitional Game (VPCG)} is a simple coalitional game $(N,v)$ such that $N=V$ and for any $C \subseteq N$, $v(C)=1$ if and only if $C$ admits an $s$-$t$ path.
	\end{itemize}	
\end{definition}

Finally, we define the class of skill games, which were recently introduced by \citet{BaRo08a}. 
\begin{definition}[Coalitional skill games \cite{BaRo08a}]
A \emph{coalitional skill domain} is composed of players $N$, a set of tasks $T=\{t_1,\ldots, t_m\}$ and a set of skills $S = \{s_1,\ldots, s_k\}$. Each player $i$ has a set of skills $S(i)\subseteq S$, and each task $t_j$ requires a set of skills $S(t_j)\subseteq S$. The set of skills a coalition $C$ has is $S(C)=\bigcup_{i\in C}S(i)$. A coalition $C$ can perform task $t_j$ if $S(t_j)\subseteq S(C)$. The set of tasks a coalition $C$ can perform is $T(C)=\{t_j\mid S(t_j)\subseteq S(C)\}$. A \emph{task value function} is a monotonic function $u:2^T\rightarrow \mathbb{R}$.	
A \emph{coalitional skill game (CSG)} in a coalitional skill domain is a game $(N,v)$ such that for all $C\subseteq N$, $v(C)=u(t(C))$. 
A \emph{weighted task skill game (WTSG)} is a CSG where each task $t_j\in T$ has a weight $w_j\in \mathbb{R}^+$ and the task value function $u(T')=\sum_{j\mid t_j\in T'}w_j$. A threshold version of WTSG can be defined according to Definition~\ref{def:threshold}.
\end{definition}

\eat{
\begin{definition}[Linear Production game]
A \emph{linear production game}~\citep{Owen75a} consist of players $N$, $m$ resources and $p$ products. Each player is give n a vector $b^i=(b_1^i,\ldots, b_m^i)$ of resources. Each unit of product $j$ requires $a_{kj}$ units of the $k$th resource and can be sold at price $c_j$. The value of a coalition $C$ is the maximum profit it can obtain with all the resources possessed by its members. 
\end{definition}
}

\begin{definition}[Linear games \cite{taylorzwicker}]
On a coalitional game $(N,v)$, we define the \emph{desirability relation} $\succeq_D$ as follows:
we say that a player $i \in N$ is \emph{more desirable than} a player $j \in N$ ($i \succeq_D j$) if for all coalitions $C \in N \backslash \{i,j\}$
we have that $v(C \cup \{i\}) \geq v(C \cup \{j\})$. The relations $\succ_D$ (``strictly more desirable''), $\sim_D$ (``equally desirable''), and $\preceq_D$ and $\prec_D$ (``(strictly) less desirable'') are defined in the obvious fashion.
\emph{Linear games} are monotonic simple games with a complete desirability relation, i.e. every pair of players is comparable with respect to $\succeq_D$.
Weighted voting games form a strict subclass of linear games.
A linear game on players $N = \{1, \ldots, n\}$ is \emph{canonical} iff $\forall i,j \in N, i< j : i \succeq_D j$.
A \emph{right-shift} of a coalition $C$ is a coalition that can be obtained by a sequence of replacements of players in $C$ by less desirable players. A \emph{left-shift} of a coalition $C$ is defined analogously.
Canonical linear games can be represented by listing their \emph{shift-minimal winning coalitions}: minimal winning coalitions for which it holds that any right-shift is losing. Similarly they can be represented by listing their \emph{shift-maximal losing coalitions}, defined as obvious.
\end{definition}

\subsection{Problem definition}
We formally define \emph{coalition structures} and {\sc OptCS}.

\begin{definition}[Optimal coalition structure]
A \emph{coalition structure} for a game $(N,v)$ is a partition of $N$. The \emph{social welfare attained by} a coalition structure $\pi$, denoted $v(\pi)$ (we overload notation), is defined as $\sum_{C \in \pi} v(C)$. A coalition structure $\pi$ is \emph{optimal} when $v(\pi) \geq v(\pi')$ for every coalition structure $\pi'$.
\end{definition}

We consider the following standard computational problem in our paper.

\begin{definition}[Problem {\sc OptCS}]
For any class of coalitional games $X$, and its associated natural representation, the problem {\sc OptCS($X$)} is as follows: given a coalitional game $(N,v) \in X$, compute an optimal coalition structure. 
\end{definition}


\section{Games with fixed player types}
\label{sec-playertypes}
We study the problem of computing an optimal coalition structure for a game in the case that the number of \emph{player types} is fixed. 
\citet{SAK10a} considered player types and showed that some intractable problems become tractable when only dealing with a fixed number of player types. They did not address coalition structure generation in their paper.

\begin{definition}[Player type]
For a coalitional game $(N, v)$, we call two players $i,j \in N$ \emph{strategically equivalent} iff for every coalition $C \in N \backslash \{i,j\}$ it holds that $v(C \cup \{i\}) = v(C \cup \{j\})$.
When two players $i,j \in N$ are strategically equivalent, we say that $i$ and $j$ are of the same \emph{player type}.
\end{definition}

\begin{definition}[Valid type-partition] 
A \emph{valid type-partition} for a game $(N, v)$ is a partition $P$ of $N$ such that for each player set $C \in P$, all players in $C$ are of the same player type.
\end{definition}
Let {\sc OptCS($k$-types)} be the problem where the goal is to compute an optimal coalition structure for a coalitional game $(N,v)$, given as input a partition $P$ of $N$ with $|P| \leq k$ and the characteristic function $v$. Note that if all players are different, then  $|P|= n$. 
In general it is not easy to verify that a given partition for a simple game is a valid type-partition. But under the assumption that we are given a valid type-partition, and $v$ is easy to compute, it turns out that an optimal coalition structure can be computed in polynomial time.

\subsection{A general algorithm}

Now we will show that there exists a general polynomial-time algorithm to compute an optimal coalition structure for any  coalitional game when we are given a valid type-partition with a number of player types bounded by a constant. 
Our algorithm utilizes dynamic programming to compute an optimal coalition structure provided there are a constant number of player types.

\begin{theorem}\label{pro:Monotone-player-types}
There is a polynomial-time algorithm for {\sc OptCS($k$-types)}, provided that querying $v$ takes at most polynomial time, and the given input partition is a valid type-partition.
\end{theorem}
\begin{proof}
Let $N = \{1, \ldots, n\}$ be the player set and $P = \{T_1, \ldots, T_k\}$ be the input type-partition.
We define \emph{coalition-types} as follows: for non-negative integers $t_1, \ldots, t_k$, the \emph{coalition-type} $T(t_1, \ldots, t_k)$ is the set of coalitions $\{C\ |\ \forall i \in \{1, \ldots, k\} : |C \cap T_i| = t_i\}$. In words, coalitions in coalition-type $T(t_1, \ldots, t_k)$ have $t_i$ players of type $T_i$, for $1 \leq i \leq k$. Note that $v$ maps all coalitions of the same coalition-type to the same value. 

First our algorithm computes a table $V$ of values for each coalition type. In order to do this we need to query $v$ at most $n^k$ times, since $1 \leq t_i \leq n$ for all $i$, $1 \leq i \leq k$. Let $\text{time}(v)$ denote the time it takes to query $v$, then computing $V$ takes $O(n^k \cdot $time$(v))$ time.

We proceed with a dynamic programming approach in order to find an optimal coalition structure: Let $f(a_1, \ldots, a_k)$ be the optimal social welfare attained by an optimal coalition structure on a game $(N', v)$ with $N' \in \{N'\ |\ \forall i \in \{1,\ldots, k\} : |N' \cap T_i| = a_i\}$. Note that it does not matter which $N'$ we choose from this set: the choice of $N'$ has no effect on the optimal social welfare since all $N'$ are of the same coalition-type. We are interested in computing $f(|T_1|, \ldots, |T_k|)$. By $\gamma(G)$, we signify those type-partitions which generate the same total utility as the empty set.

Since $v(\emptyset)=0$, the following recursive definition of $f(a_1, \ldots, a_k)$ follows:
\begin{equation}\label{eqn:dyna-fixed-types}
f(a_1, \ldots, a_k) = \begin{cases}
			0 \quad \text{ if } a_i = 0 \text{ for } 1 \leq i \leq k, \\ 
                         \max\{f(a_1 - b_1, \ldots, a_1 - b_k) + v(b_1, \ldots, b_k)\\
			 \qquad \ |\ \forall i \in \{1, \ldots, k\} : b_i \leq a_i \} \quad \text{ otherwise.}
                       \end{cases}
\end{equation}

The recursive definition of $f(a_1, \ldots, a_k)$ directly implies a dynamic programming algorithm. 
The dynamic programming approach works by filling in a $|T_1| \times \cdots \times |T_k|$ table $Q$, where the value of $f(a_1, \ldots a_k)$ is stored at entry $Q[a_1, \ldots, a_k]$. 
Once the table has been computed, $f(|T_1|, \ldots, |T_k|)$ is returned. 
The entries of $Q$ are filled in according to (\ref{eqn:dyna-fixed-types}). In order to utilize (\ref{eqn:dyna-fixed-types}), ``lower'' entries are filled in first, i.e.\ $Q[a_1, \ldots, a_k]$ is filled in before $Q[a_1', \ldots, a_k']$ if $a_i \leq a_i'$ for $1 \leq i \leq k$.
Evaluating (\ref{eqn:dyna-fixed-types}) then takes $O(n^k)$ time (due to the ``otherwise''-case of (\ref{eqn:dyna-fixed-types}), where the maximum of a set of at most $n^k$ elements needs to be computed). There are $O(n^k)$ entries to be computed, so the algorithm runs in $O(n^k \cdot $time$(v) + n^{2k})$ time. 

It is straightforward to extend this algorithm so that it (instead of outputting only the optimal social welfare) also computes and outputs an actual coalition structure that attains the optimal social welfare. 
To do so, maintain another table $|T_1| \times \cdots \times |T_k|$ table $R$. At each point in time that some entry of $Q$ is computed, say $Q[a_1, \ldots, a_k]$, now we also fill in $R[a_1, \ldots, a_k]$. $R[a_1, \ldots, a_k]$ contains a description of a set $\mathcal{C}$ of coalitions such that $\sum_{C \in \mathcal{C}} v(C) = f(a_1, \ldots, a_n)$ and $\bigcup \mathcal{C} \in T(a_1, \ldots, a_k)$. It suffices to describe $\mathcal{C}$ by simply listing the type of each $C \in \mathcal{C}$, and it is straightforward to verify that we can set $R(a_1, \ldots, a_k)$ to $\varnothing$ if $(a_1, \ldots, a_k) \in \gamma(G)$, and otherwise we set $R(a_1, \ldots, a_k)$ to $(P(a_1 - b_1, \ldots, a_1 - b_k), (b_1, \ldots, b_k))$, where $(b_1, \ldots, b_k)$ is the argument in the $\max$-expression of (\ref{eqn:dyna-fixed-types}).
\end{proof}

\subsection{Difficulty of finding types}

The polynomial-time algorithm given in Theorem \ref{pro:Monotone-player-types} relies on the promise that the type-partition given in the input is valid. A natural question is now whether it is also possible to efficiently compute the type-partition of a game in polynomial time when given only the weaker promise that the number of player types is constant $k$. We answer this question negatively. For randomized algorithms, we show high communication complexity is necessary, i.e.\ we show that an exponential amount of information is needed from the characteristic function $v$ when we are given no information on the structure of the characteristic function and we rely only on querying $v$. In fact, the theorem states that this is the case even when $v$ is simple and $k = 2$. It should be noted that this result also holds for deterministic algorithms, since they are a special case of randomized algorithms. Despite this negative result, we show in Section \ref{sec-apps} that we can do better for some subclasses of coalitional games, when we are provided information on the structure of function $v$.

\begin{theorem}\label{pro:type-partition-hard}
Any randomized algorithm that computes a player type-partition when given as input a monotonic simple game $(N,v)$ that has $2$ player types, requires at least $\Theta(\frac{2^n}{\sqrt{n}})$ queries to $v$.
\end{theorem}
\begin{proof}
We use Yao's minimax principle \cite{yaominmax}, which states that the expected cost of a randomized algorithm on a given problem's worst-case instances is at least the lowest expected cost among all deterministic algorithms that run on any fixed probability distribution over the problem instances.

Consider the following distribution over the input, where the player set is $N = \{1, \ldots, n\}$ and $n$ is even, the number of player types is always $k = 2$, and the given game $(N,v)$ is simple and monotonic. Valuation $v$ is drawn uniformly at random from the set $V = \{v_C\ |\ C \subset N, |C| = n/2 \}$ where in $v_C$, we call $C$ the \textit{critical coalition}. Function $v_C$ is specified as follows: 
\begin{itemize}
 \item $v_C(D) = 0$ when $|D| < n/2$;
 \item $v_C(D) = 1$ when $|D| > n/2$;
 \item $v_C(D) = 1$ when $D = C$, i.e.\ $D$ is the critical coalition;
 \item $v_C(D) = 0$ otherwise.
\end{itemize}

Observe that there are exactly two player types in any instance that has non-zero probability of being drawn under this distribution: when $v_C$ is drawn, the type-partition is $(C, N \backslash C)$. Also observe that for coalitions $C$ of size $\frac{n}{2}$, $v(C) = 1$ with probability $\frac{1}{\binom{n}{n/2}}$, because $v$ is drawn uniformly at random from $V$. 

Now let us consider an arbitrary deterministic algorithm $A$ that computes the type-partition for instances in this input distribution by queries to $v$. Let $C$ be the critical coalition of $n/2$ players such that $v(C) = 1$. $A$ will have to query $v(C)$ in order to know which characteristic function from $V$ has been drawn, and thus determine the type-partition correctly. Let ${Q}(v)$ be the sequence of queries to $v$ that $A$ generates. Let $\mathcal{Q}'(v)$ be the subsequence obtained by removing from $\mathcal{Q}(v)$ all queries $v(D)$ such that $|D| \not= n/2$ and all queries that occur after $v(C)$. Because $A$ is deterministic, the query sequence of $A$ is the same among all instances up to querying the critical coalition, since the critical coalitions are the only points in which the characteristic functions of $V$ differ from each other. Therefore the expected length of $\mathcal{Q}'(v)$ is $\binom{n}{n/2} / 2$. Because $A$ was chosen arbitrarily, we conclude that also the most efficient deterministic algorithm is expected to make at least $\binom{n}{n/2} / 2 = \Theta(\frac{2^n}{\sqrt{n}})$ queries to $v$, and the theorem now follows from Yao's principle.
\end{proof}

\citet{SAK10a} showed that checking whether two players are of the same type is $\mathsf{NP}$-hard for coalitional games defined by \citet{CoSa06a}. But the games are such that even computing the value of a coalition is $\mathsf{NP}$-hard. One can say something stronger.

\begin{proposition}\label{prop:coNPplayertypes}
There exists a representation of coalitional games for which checking whether two players are of the same type is $\mathsf{coNP}$-complete even if the value of each coalition can be computed in polynomial time.
\end{proposition}
\begin{proof}
A coalition $C \subseteq N \setminus \{i,j\}$ such that $v(C \cup \{i\}) \neq v(D \cup \{j\})$ is a polynomial-time certificate for membership in $\mathsf{coNP}$. Also, it is well known that checking whether two players in a WVG have the same Banzhaf index is $\mathsf{coNP}$-complete~\citep{MaMa00a}. Since two players in a WVG are of the same type if and only if they have same the Banzhaf index, we are done.
\end{proof}

\subsection{Applications of Theorem~\ref{pro:Monotone-player-types}}\label{sec-apps}
Theorem~\ref{pro:type-partition-hard} and Proposition \ref{prop:coNPplayertypes} indicate that finding player types is in general a difficult task.
Despite these negative results, Theorem~\ref{pro:Monotone-player-types} still applies to all classes of games and many natural settings where the type-partition is implicitly or explicitly evident:
\begin{corollary}\label{cor-wvg-values}
There exists a polynomial-time algorithm that solves {\sc OptCS(WVG)} in the following cases: 1.) in the input game (given in weighted form), the number of distinct weights is constant; 2.) in the input game (given in weighted form) the number of distinct weight vectors for the players is constant.
\end{corollary}
\begin{proof}
When two players have the same weight (in the case of WVGs) or weight vectors (in the case of MWVGs), they are strategically equivalent. Therefore we can type-partition the players according to their weights and apply Theorem~\ref{pro:Monotone-player-types}.
\end{proof}

There exists a polynomial-time algorithm for computing the desirability classes, when given the list of shift-minimal winning coalitions of a linear game~\citep{Aziz08a}. This immediately yields the following corollary:
\begin{corollary}
In the following cases, there exists a polynomial-time algorithm that computes an optimal coalition structure for linear games with a constant number of desirability classes: 1.) the input game is represented as a list of (shift-)minimal winning coalitions; 2.) the input game is represented as a list of (shift-)maximal losing coalitions;
\end{corollary}

\citet{BMJK10a} proved that {\sc OptCS(CSG)} is polynomial-time solvable if the number of tasks is constant and the `skill graph' has bounded tree-width. As a corollary of Theorem~\ref{pro:Monotone-player-types}, we obtain a complementing positive result which applies to all of the coalitional skill games defined in \citep{BaRo08a}. 

\begin{corollary}\label{cor:scsg-types}
There exists a polynomial-time algorithm that computes an optimal coalition structure for WTSGs and T-WTSGs with at most a fixed number of player types or a fixed number of skills.
\end{corollary}
\begin{proof}
Assume that there the number of skills is a constant $k'$. Then there is a maximum of $2^{k'}$ player types. A polynomial-time algorithm that computes an optimal coalition structure now follows from Theorem~\ref{pro:Monotone-player-types}.
\end{proof}

\section{Weighted voting games and simple games}
\label{sec-wvg-and-simple}
In this section, we examine weighted voting games (WVGs) and, more generally, simple games. Weighted voting games are coalitional games widely used in multiagent systems and AI. We have already seen that there exists a polynomial-time algorithm to compute an optimal coalition structure for WVGs with a constant number of weight values. We show that if the number of weight values is not a constant, then the problem becomes strongly $\mathsf{NP}$-hard. 

\begin{proposition}\label{pro:wvg-hard}
For a WVG, checking whether there is a coalition structure that attains social welfare $k$ or more is $\mathsf{NP}$-complete.
\end{proposition}
\begin{proof}
We prove this by a reduction from an instance of the classical $\mathsf{NP}$-hard {\sc Partition} problem to checking whether a coalition structure in a WVG gets social welfare at least $2$. 
An instance of the problem $k$-{\sc Partition} is a set of $n$ integer weights $A=\{a_1, \ldots, a_n \}$ and the question is whether it is possible to partition $A$, into $k$ subsets $P_1\subseteq A$,\dots $P_k\subseteq A$ such that $P_i\cap P_j=\varnothing$ and $\bigcup_{1\leq i\leq k}P_i=A$ and for all $i\in \{1,\ldots, k\}$, $\sum_{a_j\in A_i}a_j=\sum_{1\leq j\leq n}a_j/k$.

Without loss of generality, assume that $W=\sum_{a_i\in A}a_i$ is a multiple of $k$. Given an instance of $k$-{\sc Partition} $I=\{a_1, \ldots, a_k\}$, we can transform it to a WVG $v=[q;w_1,\ldots, w_k]$ where $w_i=a_i$ for all $i\in \{1,\ldots, k\}$ and $q=W/k$. Then the answer to $I$ is yes if and only if there exists a coalition structure $\pi$ for $v$ such that $v(\pi)=k$. 
\end{proof}

Since $3$-{\sc Partition} is strongly $\mathsf{NP}$-complete, it follows that {\sc OptCS(WVG)} is strongly $\mathsf{NP}$-hard. 
This is contrary to the other results concerning WVGs where computation becomes easy when the weights are encoded in unary~\citep{MaMa00a}. 
Note that any strongly $\mathsf{NP}$-hard optimization problem with a polynomially bounded objective function cannot have an $\mathsf{FPTAS}$ unless $\mathsf{P} = \mathsf{NP}$. 
Proposition~\ref{pro:wvg-hard} does not discourage us from seeking an approximation algorithm for WVGs. We show that there exists a 2-optimal polynomial-time approximation algorithm:
\begin{proposition}\label{pro:WVG-approx}
There exists a 2-optimal polynomial-time approximation algorithm for {\sc OptCS(WVG)}.
\end{proposition}
\begin{proof}
Consider the following algorithm: 
Let $[q; w_1, \ldots, w_n]$ be the input (so $N = \{ 1, \ldots, n \}$). We assume without loss of generality that $w_i \leq q$ for all $i$.
The algorithm first sets $p[0] := 0$, and then computes for some number $c$ the values $p[1], \ldots, p[c]$ using the rule 
\begin{equation}
p[i] := 
\begin{cases}
n \quad \textrm{  if  } \sum_{k = p[i-1] + 1}^{n} w_k < q ,\\
\min \{j\ |\ \sum_{k = p[i-1] + 1}^{j} w_k \geq q, (p[i-1]+1) \leq j \leq n \} \quad \textrm{ otherwise, }
\end{cases}
\end{equation}
where $c$ is taken such that $p[c] = n$.
The algorithm outputs the coalition structure $\{C_1, \ldots, C_c\}$, where for $1 \leq i \leq c$, $C_i = \{p[i-1]+1, \ldots, p[i]\}$.

Observe that the coalitions $C_1$ to $C_{c-1}$ are all winning and $C_c$ is not necessarily winning, so the value of the computed coalition structure is at least $c-1$
By our assumption, the total weight of any of the coalitions $C_1, \ldots, C_{c-1}$ is less than $2q$, and the total weight of $C_c$ is less than $q$.
Therefore, the total weight of $N$ is strictly less than $q(2c-1)$, so the optimal social welfare is at most $2c-2 = 2(c-1)$. This is two times the social welfare of the coalition structure computed by the algorithm.
\end{proof}
A tight example for the algorithm described in the proof of Theorem \ref{pro:WVG-approx} would be $[q; q-\epsilon, q-\epsilon, \epsilon, \epsilon]$, where $q$ is a fixed constant and $\epsilon$ is any positive real number strictly less than $q/2$. On this input, the algorithm outputs a coalition structure that attains a social welfare of 1, while the optimal social welfare is clearly 2. 
The following proposition shows that there does not exist a better polynomial-time approximation algorithm under the assumption that $\mathsf{P}\not=\mathsf{NP}$.

\begin{proposition}\label{pro:WVG-inapprox}
Unless $\mathsf{P}=\mathsf{NP}$, there exists no polynomial-time algorithm which computes an $\alpha$-optimal coalition structure for a WVG where $\alpha<2$.
\end{proposition}
\begin{proof}
We would be able to solve the $\mathsf{NP}$-complete problem {\sc Partition} in polynomial time if there existed a ($<2$)-optimal polynomial-time approximation algorithm for {\sc OptCS(WVG)}. We could reduce a partition instance $(w_1, \ldots, w_n)$ to a weighted voting game $[q; w_1, \ldots, w_n]$ where $q = \frac{\sum_{i=1}^n w_n}{2}.$ 
Because the sum of all weights of the players is $2q$, a ($<2$)-optimal approximation algorithm would output an optimal coalition structure when provided with this instance. The output coalition structure directly corresponds to a solution of the original {\sc Partition} instance, in case it exists. Otherwise, the social welfare attained by the output coalition structure is 1.
\end{proof}

Simple games that are not necessarily weighted, and are represented by the list of minimal winning coalitions, are even harder to approximate.

\begin{proposition}\label{pro:wm-hard}
{\sc OptCS(MWC)}, i.e.\ {\sc OptCS} for simple games represented as a list of minimal winning coalitions, cannot be approximated within any constant factor unless $\mathsf{P}=\mathsf{NP}$.
\end{proposition}
\begin{proof}
	This can be proved by a reduction from an instance of the classical $\mathsf{NP}$-hard maximum clique ({\sc MaxClique}) problem. It is known that {\sc MaxClique} cannot be approximated within any constant factor~\citep{MaGa00a}.

Consider the instance $I$ of {\sc MaxClique} represented by an undirected graph $G_I=(V,E)$. Transform $I$ into instance $I' = (N, W^m)$ of {\sc OptCS(MWC)} in the following way. Define $N = \{\{v,v'\} : v \in V, v' \in V\}$ to be all subsets of $V$ of cardinality 2. Next, set $W^m=\{C_i:i\in V\}$, and for all $i \in V$ define $C_i = \{\{i,j\}\ |\ \{i,j\} \not\in E\}$. Now two coalitions $C_i$ and $C_j$ are disjoint if and only if $\{i,j\} \in E$. Then the maximum clique size is greater than or equal to $k$ if and only if there is a coalition structure for $(N, W^m)$ that attains social welfare $k$. Now assume that there exists a polynomial-time algorithm which computes a coalition structure $\pi$ which gets social welfare within a constant factor $\alpha$ of the maximum possible social welfare $k$. Then we can use $\pi$ to get a constant-factor approximation solution to instance $I$ in polynomial time in the following way. Consider the set of vertices $\{i:C_i\in \pi\}$. Since for $C_i, C_j\in \pi$, $C_i$ and $C_j$ are disjoint, then we know that $(i,j)\in E$. Therefore the vertices $\{i:C_i\in \pi\}$ form a clique of size $k/\alpha$. 
\end{proof}

\section{Games on graphs}\label{sec-graph-games}
Numerous classes of coalitional games are based on graphs. We characterize the complexity of {\sc OptCS} for many of these classes in the section. We first turn our attention to one such class for which the computation of cooperative game solutions is well studied~\citep{DePa94a}. We see that that {\sc OptCS} is computationally hard in general for graph games:

\begin{proposition}\label{pro:GG-hard}
For the general class of graph games GG, the problem {\sc OptCS} is strongly $\mathsf{NP}$-hard.
\end{proposition}
\begin{proof}
We prove by presenting a reduction from the strongly $\mathsf{NP}$-hard problem {\sc MaxCut}. Consider an instance $I$ of  {\sc MaxCut} with a connected undirected graph $G = (V,E,w)$ and non-negative weights $w(i,j)$ for each edge $(i,j)$. Let $W=\sum_{(i,j)\in E} w(i,j)$ and define $P(i)$ as the vertices on the same side as as vertex $i$. We show that if there is a polynomial-time algorithm which computes an optimal coalition structure, then we have a polynomial-time algorithm for {\sc MaxCut}. There exists a polynomial-time reduction that reduces $I$ to an instance $I'=(V',E',w')$ of {\sc OptCS} for graph games where $V'=V\cup\{x_1,x_2\}$ and $E'= E\cup\{\{x_1,i\}:i\in N\}\cup \{\{x_2,i\}:i\in N\}\cup \{\{x_1,x_2\}\}$. The weight function $w'$ is defined as follows: $w'(a,b)=-w(a,b)$ if $a, b\in V$, $w'(a,b)=W+1$ if $a\in \{x_1,x_2\}$ and $b\in V$, $w'(a,b)=-(|V|+1)W$ if $a=x_1$ and $b=x_2$.
	
We now show that a solution to instance $I'$ of {\sc OptCS(GG)} can be be used to solve instance $I$ of {\sc MaxCut}. Assume that $\pi'$ is an optimal coalition structure for $I'$. Then we know that $\pi$ is of the form $\{\{x_1, A'\}, \{x_2, B'\}\}$ where $(A',B')$ is a partition of $V$. We also know that $\sum_{a \notin \pi'(b)}w'(a,b)$ is minimized in $\pi'$. Therefore, we have a corresponding partition $\pi$ of $V$ such that $\sum_{a\notin \pi(b)}w(a,b)$ is maximized.   
\end{proof}

\begin{observation}\label{obs:gg+-easy}
It is clear that for GG$^{+}$, the coalition structure containing only the grand coalition is the optimal coalition structure.
\end{observation}

We now present some positive results concerning {\sc OptCS} for other games on graphs:

\begin{proposition}\label{pro:SCG-easy}
{\sc OptCS}(SCG) can be solved in polynomial time. 
\end{proposition}
\begin{proof}
For a SCG, {\sc OptCS} is equivalent to computing the maximum number of edge disjoint spanning subgraphs. Clearly, the maximum number of edge disjoint spanning trees is greater than or equal to the maximum number of spanning subgraphs. Since the spanning trees are also spanning subgraphs, the problem reduces to computing the maximum number of disjoint spanning trees. The problem is solvable in $O(m^2)$~\citep{RoTa85a}.
\end{proof}

\begin{proposition}\label{pro:EPCG-easy}
For EPCGs and VPCGs, {\sc OptCS} can be solved in polynomial time.
\end{proposition}
\begin{proof}
The problems are equivalent to computing the maximum number of edge disjoint and vertex disjoint $s$-$t$ paths respectively.
There are well-known algorithms to compute them. For example, the maximum number of edge-disjoint $s$-$t$ paths is equal to the max flow value of the graph in which each edge has unit capacity. 
The problem of maximizing the number of of vertex disjoint paths can be reduced to maximizing the number of of vertex disjoint paths in the following way: duplicate each vertex (apart from $s$ and $t$) with one getting all ingoing edges, and the other getting all the outgoing edges, and an internal edge between them with the node weight as the edge weight.
\end{proof}

\begin{proposition}\label{pro:NFG-matching-easy}
The coalition structure containing only the grand coalition is an optimal coalition structure for: 1.) NFGs and 2.) Matching games.
\end{proposition}
\begin{proof}
1.) Assume there is a coalition structure $\pi$ of the edges which achieves the total social welfare of $s$. This means that the sum of the net flow for each $E'\in \pi$ totals $s$. Since each member of $\pi$ is mutually exclusive, for any $A,B\in \pi$, the flows in $A$ and $B$ do not interact with each other. Now, consider the coalition structure $\pi'=\{E\}$ which consists of the grand coalition. Then $E$ can achieve a network flow of at least $s$ by having exactly the same flows as that of $\pi$, we know that $v(\pi')\geq s$. Therefore, the coalition structure consisting of only the grand coalition attains a social welfare that is at least the social welfare attained by any other coalition structure. 

2.) Assume there is a coalition structure $\pi=\{V_1,\ldots, V_k\}$ of the vertices that attains a social welfare of $s$. Let the maximum weighted matching of the graph $G[V_i]$ restricted to vertices $V_i$ be $m_i$. Then we know that $\sum_{1\leq i\leq k} m_i=s$. Since each member of $\pi$ is mutually exclusive, for any $V_i,V_j\in \pi$, the matchings in $G(V_i)$ and $G(V_j)$ have no intersection with each other. Now, consider the coalition structure $\pi'=\{E\}$ which consists of the grand coalition. Then $V$ can achieve a maximum matching of at least $s$ by having exactly the same matchings as that of vertex sets in $\pi$. This implies that that $v(\pi') \geq s$. Therefore, the coalition structure consisting of only the grand coalition attains a social welfare that is at least the social welfare attained by any other coalition structure.
\end{proof}

On the other hand, the threshold versions of certain games are computationally harder to solve because of their similarity to WVGs~\citep{ABH10a}. As a corollary of Prop.~\ref{pro:WVG-inapprox}, we obtain the following:
\begin{corollary}\label{pro:T-NFG-matching-hard}
	Unless $\mathsf{P}=\mathsf{NP}$, there exists no polynomial-time algorithm which computes an $\alpha$-optimal coalition structure for $\alpha<2$ and for the following classes of games: 1. T-NFG. 2. T-Matching game and 3. T-GG$^+$.
	\eat{
	\begin{enumerate}
		\item T-NFG
		\item T-Matching game
		\item T-GG$^+$
	\end{enumerate}
	}
\end{corollary}
\eat{
\begin{proof}
We prove $\mathsf{NP}$-hardness by a reduction from {\sc OptCS(WVG)} for each of the respective games.
\begin{enumerate}
\item Take any WVG $v_I=[q;w_1,\ldots, w_n]$. Create a corresponding threshold matching game $G_I$. Game $G_I$ consists of threshold $q$ and weighted graph $(V,E,w)$ where $V=\{s,t\}$ and $E=\{1,\ldots, n\}$ consists of $n$ parallel edges with capacity $w_i$ for edge $i$. 
Then we know that a partition of players in $v_I$ gets social welfare $x$ if and only a partition of edges $E$ gets social welfare $x$. 
Since we already know the problem (of even approximating {\sc OptCS}) is $\mathsf{NP}$-complete for WVGs, it follows that the problem is also $\mathsf{NP}$-complete for T-NFGs.
\item Take any WVG $v_I=[q;w_1,\ldots, w_n]$. Create a corresponding threshold matching game $G_I$. Game $G_I$ consists of threshold $q$ and weighted graph $(V,E,w)$ where $|V|=2n$, $V = \{v_1, \ldots, v_{2n}\}$, $|E|=n$ and $E=\{\{v_{2i-1}, v_{2i}\}:i=1,\ldots, n\}$ such that $w(\{v_{2i-1}, v_{2i}\})=w_i$. Consider a coalition structure $\pi=\{A,B\}$ for WVGs. Then there is a corresponding coalition structure $\pi'$ for T-Matching game where $\pi'=\{\{v_{2i-1}, v_{2i}:i\in A\},\{v_{2i-1}, v_{2i}:i\in B\}\}$. We then know that coalition structure $\pi$ for $v_I$ attains a social welfare of $x$  if and only if partition $\pi'$ for $G_I$ attains a social welfare of $x$. 
\item We can use the same reduction as for T-Matching games.
\end{enumerate}
\end{proof}
}
In some cases, {\sc OptCS} may be expected to be intractable because the coalitional game is defined on a combinatorial optimization domain which itself is intractable. We observe that even if computing the value of coalitions is intractable, solving {\sc OptCS} may be easy:
\begin{observation}\label{obs:indep-set-game}
Given an instance of maximum independent set, graph $G=(V, E)$, finding the value of the coalition $v(N)$ is $\mathsf{NP}$-hard, but the optimal coalition structure is all singletons.\end{observation}

\section{Conclusions}
\label{sec-conc}

\begin{table*}[t]
\centering
\scriptsize 
\begin{tabular}{ll}
\toprule

Game class&Complexity of {\sc OptCS}\\ \midrule
Coalition value oracle (valid type-partition \& const. \#types)&$\mathsf{P}$~(Th. \ref{pro:Monotone-player-types})\\ 

WVG (const no. weight values)&$\mathsf{P}$~(Cor. \ref{cor-wvg-values})\\

(T-)WTCSGs (const. \#skills or const. \#types)&$\mathsf{P}$~(Cor.~\ref{cor:scsg-types})\\

WCSG (const. \#tasks, bounded tree-width skill graph)&$\mathsf{P}$~\citep{BMJK10a}\\

SCG&$\mathsf{P}$~(Prop.~\ref{pro:SCG-easy})\\ 

EPCG and VPCG&$\mathsf{P}$~(Prop.~\ref{pro:EPCG-easy})\\

NFG and Matching Game&$\mathsf{P}$~(Prop.~\ref{pro:NFG-matching-easy})\\

Marginal Contribution Nets&$\mathsf{NP}$-hard~\citep{0CI+09a}\\ 

GG$^{+}$&$\mathsf{P}$~(Obs.~\ref{obs:gg+-easy})\\

Independent Set Game&$\mathsf{P}$~(Obs.~\ref{obs:indep-set-game})\\

GG& Strongly $\mathsf{NP}$-hard (Prop.~\ref{pro:GG-hard})\\

$(N,W^m)$&$\mathsf{NP}$-hard to approx. within const. factor (Prop.~\ref{pro:wm-hard})\\
WVG&Strongly $\mathsf{NP}$-hard (Prop.~\ref{pro:wvg-hard}); \\ & $\mathsf{NP}$-hard to approx. within factor $<2$ (Prop.~\ref{pro:WVG-inapprox})\\ 

T-Matching; T-NFG; T-GG&$\mathsf{NP}$-hard to approx. within factor $<2$  (Cor.~\ref{pro:T-NFG-matching-hard})\\  

CSG&$\mathsf{NP}$-hard even for SCSGs~\citep{BMJK10a}\\

\bottomrule
\end{tabular}
\caption{Summary of complexity results for {\sc OptCS}}
\label{OptCS-summary}
\end{table*}

Coalition structure generation is an active area of research in multiagent systems. 
We presented a general positive algorithmic result for coalition structure generation, namely that an optimal coalition structure can be computed in polynomial time if the player types are known and the number of player types is bounded by a constant. 
In many large multiagent systems, it is a valid assumption that there are a lot of agents but the agents can be divided into a bounded number of strategic classes. For example, skill games are well motivated for coordinated rescue operation settings~\citep{BaRo08a,BMJK10a}. In these settings, there may be a large number of rescuers but they can be divided into a constant number of types such as firemen, policemen and medics. 
We have also undertaken a detailed study of the complexity of computing an optimal coalition structure for a number of well-studied games and well-motivated games in AI, multiagent systems and operations research. The results are summarized in Table~\ref{OptCS-summary}.

\section*{Acknowledgements}
We thank Hans Georg Seedig and the anonymous referees for helpful feedback.

\bibliography{main}

\end{document}